\documentclass[a4paper,UKenglish]{lipics-v2016}
\usepackage{microtype}%if unwanted, comment out or use option "draft"
\usepackage[noend]{algpseudocode}
\usepackage{algorithm}
\title{Competitive Routing in Hybrid Communication Networks\footnote{This work was partially supported by the German Research Foundation (DFG) within the Collaborative Research Center 'On-The-Fly Computing' (SFB 901).}}
\titlerunning{Competitive Routing in Hybrid Communication Networks} %optional, in case that the title is too long; the running title should fit into the top page column

\author[1]{Daniel Jung}
\author[2]{Christina Kolb}
\author[3]{Christian Scheideler}
\author[4]{Jannik Sundermeier}
\affil[1]{Heinz Nixdorf Institute \& Computer Science Department, Paderborn University, Paderborn, Germany\\
  \texttt{jungd@hni.upb.de}}
\affil[2]{Department of Computer Science, Paderborn University, Paderborn, Germany\\
	\texttt{ckolb@mail.uni-paderborn.de}}
\affil[3]{Department of Computer Science, Paderborn University, Paderborn, Germany\\
	\texttt{scheideler@uni-paderborn.de}}
\affil[4]{Department of Computer Science, Paderborn University, Paderborn, Germany\\
  \texttt{janniksu@mail.uni-paderborn.de}}
\authorrunning{D. Jung, C. Kolb, C. Scheideler and J. Sundermeier} %mandatory. First: Use abbreviated first/middle names. Second (only in severe cases): Use first author plus 'et. al.'

\Copyright{Daniel Jung, Christina Kolb, Christian Scheideler and Jannik Sundermeier}%mandatory, please use full first names. LIPIcs license is "CC-BY";  http://creativecommons.org/licenses/by/3.0/

% mandatory: Please choose ACM 1998 classifications from http://www.acm.org/about/class/ccs98-html . E.g., cite as "F.1.1 Models of Computation". 
% mandatory: Please provide 1-5 keywords
\EventEditors{John Q. Open and Joan R. Acces}
\EventNoEds{2}
\EventLongTitle{42nd Conference on Very Important Topics (CVIT 2016)}
\EventShortTitle{CVIT 2016}
\EventAcronym{CVIT}
\EventYear{2016}
\EventDate{December 24--27, 2016}
\EventLocation{Little Whinging, United Kingdom}
\EventLogo{}
\SeriesVolume{42}
\ArticleNo{23}
\newcommand{\localDel}[0]{$LDel^{2}(V)$}

\begin{document}

\maketitle
\begin{abstract}
Routing is a challenging problem for wireless ad hoc networks, especially when the nodes are mobile and spread so widely that in most cases multiple hops are needed to route a message from one node to another. In fact, it is known that any online routing protocol has a poor performance in the worst case, in a sense that there is a distribution of nodes resulting in bad routing paths for that protocol, even if the nodes know their geographic positions and the geographic position of the destination of a message is known. The reason for that is that radio holes in the ad hoc network may require messages to take long detours in order to get to a destination, which are hard to find in an online fashion.

In this paper, we assume that the wireless ad hoc network can make limited use of long-range links provided by a global communication infrastructure like a cellular infrastructure or a satellite in order to compute an abstraction of the wireless ad hoc network that allows the messages to be sent along near-shortest paths in the ad hoc network. We present distributed algorithms that compute an abstraction of the ad hoc network in $\mathcal{O}\left(\log ^2 n\right)$ time using long-range links, which results in $c$-competitive routing paths between any two nodes of the ad hoc network for some constant $c$ if the convex hulls of the radio holes do not intersect. We also show that the storage needed for the abstraction just depends on the number and size of the radio holes in the wireless ad hoc network and is independent on the total number of nodes, and this information just has to be known to a few nodes for the routing to work.
\end{abstract}

\hfill \\
\noindent
This paper is eligible for best student paper award (all of the authors are full-time students).

\section{Introduction}
Nowadays almost every person has a cell phone. Hence, in a city center the density of cell phones would, in principle, be sufficiently high to set up a well-connected wireless ad hoc network spanning the entire city center, which could then be used for many interesting applications in the area of social networks. Wireless ad hoc networks have the advantage that there is no limit (other than the bandwidth and battery constraints) on the amount of data that can be exchanged while the amount of data that can be transferred at a reasonable rate via long-range links using the cellular infrastructure or satellite is limited (by some data plan) or costly. However, routing in a mobile ad hoc network is challenging, even if the geographic position of the destination is known, since buildings or other obstacles like rivers may create radio holes that make it non-trivial to find a near-shortest routing path. So the question we address in this paper is:

Can long-range links be used effectively to find near-shortest routing paths in the ad hoc network?

A simple solution to that problem would be that all nodes regularly post their geographic position and the nodes within their communication range to a server in the Internet. This would allow the server to compute optimal routing paths so that whenever a node wants to forward a message to a certain destination, the server can tell it which of the neighbors to send it to. An alternative approach that we are pursuing in this paper is a purely peer-to-peer based approach in which no other equipment other than the cell phones and an infrastructure for the long-range links needs to be used. To the best of our knowledge, our approach is the first one that is making use of a global communication infrastructure in a peer-to-peer manner in order to efficiently determine short routing paths for an ad hoc network. Wireless ad hoc networks have been considered before that utilize base stations in order to exchange messages more effectively, but there, messages will be sent via long-range links to bridge long distances while we will only allow messages to be sent via ad hoc links.

\subsection{Model}

\begin{figure*}[h]
	\centering
	\includegraphics[width=\textwidth]{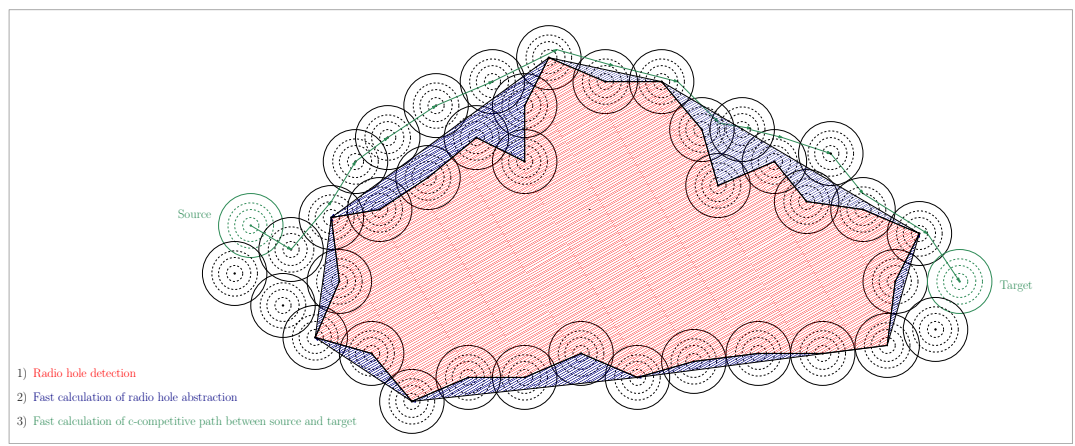}
	\caption{An overview of our approach. It contains the detection of radio holes (1), the computation of a hole abstraction (2) and a routing algorithm that finds $c$-competitive paths (3). The blue regions are called \emph{bay areas}.}
	\label{figure:overview}
\end{figure*}
\newpage
Throughout this paper, we consider $V \subset \mathbb{R}^2$ to be a set of nodes in the Euclidean plane with unique IDs (e.g., phone numbers), where $\vert V\vert =n$.
For any given pair of nodes $u=\left(u_x,u_y\right)$, $v=\left(v_x,v_y\right)$, we denote the Euclidean distance between $u$ and $v$ by $\vert\vert uv \vert \vert =$  $\sqrt{\left(u_x-v_x\right)^2 + \left(u_y-v_y\right)^2}$.
We model our cell phone network as a hybrid directed graph $H=\left(V,E,E_{AH}\right)$ where the node set $V$ represents the set of cell phones, an edge $\left(v,w\right)$ is in $E$ whenever $v$ knows the phone number (or simply {\em ID}) of $w$, and an edge $\left(v,w\right) \in E$ is also in the {\em ad hoc edge set} $E_{AH}$ whenever $v$ can send a message to $w$ using its Wifi interface.
For all edges $\left(v,w\right) \in E \setminus E_{AH}$, $v$ can only use a long-range link to directly send a message to $w$.
We adopt the {\em unit disk graph model} for the edges in $E_{AH}$.

\begin{definition}
	For any point set $V\subseteq \mathbb{R}^2$ the \emph{Unit Disk Graph} of $V$, $\mathrm{UDG}\left(V\right)$, is a bi-directed graph that contains all edges $\left(u,v\right)$ with $\vert\vert uv\vert \vert\leq 1$.
\end{definition}

We assume $\mathrm{UDG}\left(V\right)$ to be strongly connected so that a message can be sent from every node to every other node in $V$ by just using ad hoc edges.
While the ad hoc edges are fixed in Sections~\ref{section:preliminaries} - \ref{section:concreteProtocol} , the nodes can nevertheless change $E$ over time:
If a node $v$ knows the IDs of nodes $w$ and $w'$, then it can send the ID of $w$ to $w'$, which adds $\left(w,w'\right)$ to $E$.
This procedure is called \emph{ID-introduction}.
Alternatively, if $v$ deletes the address of some node $w$ with $\left(v,w\right) \in E$, then $\left(v,w\right)$ is removed from $E$.
There are no other means of changing $E$, i.e., a node $v$ cannot learn about an ID of a node $w$ unless $w$ is in $v$'s UDG-neighborhood or the ID of $w$ is sent to $v$ by some other node. \\
Moreover, we consider synchronous message passing in which time is divided into rounds.
More precisely, we assume that every message initiated in round $i$ is delivered at the beginning of round $i+1$, and a node can process all messages in a round that have been delivered at the beginning of that round.

\subsection{Objective} \label{section:objective}

Our objective is to design an efficient routing algorithm for ad hoc networks, where the source $s$ of a message knows the ID of the destination $t$, or in other words, $\left(s,t\right) \in E$.
This is a reasonable constraint since cell phone users normally wouldn't call cell phones whose users are unknown to them. Thus, whenever a message needs to be sent from a source $s$ to some destination $t$, we assume that the geographic position of $t$ is known, since $s$ can ask $t$ via a long-range link for $t's$ geographic position before sending the message towards $t$ using the ad hoc network.

Our routing algorithm consists of two parts:
After determining the radio holes of the wireless ad hoc network, we compute an abstraction, i.e., a compact representation of these radio holes and use that abstraction in order to route messages along $c$-competitive paths.
See Figure~\ref{figure:overview} for a visual description of these parts. \\
We call a routing strategy {\em $c$-competitive} if for all node pairs $\left(s,t\right)$, the routing path $(s, \dots, t)$ from $s$ to $t$ obtained by the strategy satisfies $||(s, \dots, t)||\leq c\cdot d\left(s,t\right)$, where $||(s, \dots, t)||$ denotes the Euclidean length of $(s, \dots, t)$ and $d\left(s,t\right)$ denotes the shortest Euclidean length of a path in $\mathrm{UDG}\left(V\right)$ from $s$ to $t$.

We will focus on computing suitable abstractions of radio holes in the ad hoc network. The intuition behind that is simple: if there are no radio holes, then simple greedy routing (i.e., always take the neighbor that is closest to the destination) would already give us short routing paths to arbitrary destinations. Radio holes can be specified by the nodes along its boundary, but there can be many such nodes. Therefore, we will also look at more compact representations of radio holes like the (nodes forming the) convex hull of its boundary. Considering convex hulls as radio hole abstractions makes sense because in huge cities like New York City the shape of radio holes (caused by obstacles like buildings) is in many cases convex or close to a convex shape, and these shapes do not overlap. In order to obtain the desired abstraction, we will make use of ID-introductions in order to form an overlay network that allows us to compute these abstractions in a distributed manner using the long-range links. Since sending messages via long-range links is costly (in terms of money), our goal is to keep the long-range communication work of the nodes as low as possible.

\subsection{Our Contributions}

We consider any hybrid graph $G = (V,E,E_{AH})$ where the Unit Disk Graph of $V$ is connected. Let $H$ be the set of radio holes in $G$ and $C$ be the set of convex hulls of radio holes in $H$. $P(h)$ denotes the length of the perimeter of a radio hole $h \in H$. Further, $L(c)$ denotes the circumference of a minimum bounding box of a convex hull $c \in C$. Our main contribution is:

\begin{theorem} \label{theorem:mainTheorem}
	For any distribution of the nodes in $V$ that ensures that UDG($V$) is connected and of bounded degree and that the convex hulls of the radio holes do not overlap, our algorithm computes an abstraction of UDG($V$) in $\mathcal{O}(\log^2 n)$ communication rounds using only polylogarithmic communication work at each node so that $c$-competitive paths between all source-destination pairs can be found in an online fashion.
	
	The space needed by the convex hull nodes of the radio holes is \\ $\mathcal{O}\left( \sum_{c \in C} L(c)\right)$. Nodes lying on the boundary of radio holes need storage of size $\mathcal{O}\left( \max_{h \in H} P(h) \right)$. For every other node, the space requirement is constant.
\end{theorem}

The rest of this paper is dedicated to the proof of Theorem \ref{theorem:mainTheorem}. For that, we use the following approach:
\begin{enumerate}
	\item Given the Unit Disk Graph, 
	%to determine the radio holes of the wireless ad hoc network, 
	we compute the $2$-localized Delaunay Graph. This only needs $\mathcal{O}(1)$ communication rounds. The 2-localized Delaunay Graph allows the nodes to detect whether they are at the boundary of a radio hole. Nodes at the boundary can then form a ring.
	%	Afterwards, we describe a technique that allows each node to detect whether it lies on the boundary of a hole in $\mathcal{O}(\log n)$ communication rounds.
	\item 
	%We use convex hulls as abstractions for holes. We assume that convex hulls of holes do not intersect.
	We then develop a distributed algorithm that computes a convex hull of a ring of $n$ nodes in expected $\mathcal{O}\left(\log n\right)$ communication rounds.
	\item Afterwards, we introduce the nodes of the convex hulls to each other so that they form a clique. This will allow them to compute $c$-competitive paths for all source-destination pairs that are outside of a convex hull. The introduction requires $\mathcal{O}(\log^2 n)$ communication rounds. In order to handle the case that the source or destination lies inside a convex hull, we will make use of a Dominating Set of the boundary nodes. The Dominating Set can be computed in $\mathcal{O}(\log n)$ communication rounds.
\end{enumerate}
Finally, we also consider the dynamic scenario (i.e., UDG($V$) changes over time) in Section~\ref{section:dynamicScenario}.

\subsection{Related Work}

Many routing protocols have already been proposed for wired as well as wireless communication networks. Since the focus of this paper is on wireless networks, we will restrict our overview on related work to this area.
Most research on routing in wireless networks has been done in the context of mobile ad-hoc networks, i.e., wireless networks that do not rely on some external infrastructure for their operation.
There are basically two types of routing protocols for mobile ad hoc networks: table driven / proactive protocols, and on demand / reactive protocols.
In proactive protocols, the nodes keep updating their routing tables to maintain the latest view of the network. Examples of such protocols are DSDV, OLSR, and WARP. In reactive protocols, routes are only created when required. Examples of reactive protocols are DSR, AODV, and TORA. Also hybrid protocols, i.e., protocols that are a combination of proactive and reactive approaches, are known like SRP and ZRP. See, for example, \cite{gupta2011review} for a survey.

Proactive protocols only work well for small and relatively static ad hoc networks since otherwise the routing tables can get very large, and therefore the overhead of continuously updating them becomes prohibitively expensive.
To reduce these problems, hierarchical and cluster-based strategies have been proposed like FSR and CGSR, but if the diameter of the ad hoc network is large, it may still take a long time for the tables to be up-to-date. Reactive protocols can produce a significantly lower overhead if messages are sent to only a small set of nodes.
However, if many messages with different destinations are injected, the overhead of discovering routes can become prohibitively expensive.
To reduce this problem, location-based variants have been proposed like LAR, but these only work well if the radio holes are not too large.

Routing work in theory has mostly focused on approaches where routing paths do not have to be set up before sending out a message. Instead the focus has been on simple online routing strategies that are potentially based on a suitable overlay network consisting of a subset of the wireless connections available to the nodes. The most simple online strategy is to use a greedy strategy to route a message to a destination $t$: always forward the message to the neighbor closest to~$t$ (with respect to some metric). Unfortunately, greedy strategies like Compass routing ~\cite{10.1007/3-540-40996-3_5} fail for graphs with radio holes, i.e., they might get stuck at a dead end. This can be avoided with the help of suitable virtual coordinates for the nodes (e.g., \cite{7236907}), but computing these is quite expensive. Instead, Kuhn et al. ~\cite{kuhn2003worst} proposed GOAFR, a routing strategy that uses a combination of greedy and face routing, which can find paths with quadratic competitiveness~\cite{kuhn2003worst}. They also proved that this is worst-case optimal.
That is, it is not possible to design routing strategies which use only local knowledge and achieve a better competitiveness than quadratic. Their lower bound is based on the fact that a radio hole might have a complex structure, like a maze, making it hard to find a short path to a destination in an online fashion. Some other examples of the many routing strategies that have been proposed are \cite{rao1993robot,lumelsky1987algorithmic,ruhrup2006online}.
For example, R\"uhrup and Schindelhauer considered routing strategies for grids that contain failed nodes~\cite{ruhrup2006online}.
This is similar to our scenario as failed nodes behave like radio holes in an ad hoc network.
Their procedure uses a strategic search which distributes a message over multiple paths. They proved that their procedure is asymptotically optimal for their setting.
However, it is not clear how the strategy can be generalized to arbitrary node distributions.

So the question arises: How to make use of long-range links to find a suitable abstraction of the radio holes with a local strategy such that we obtain $c$-competitive paths for any source-destination pair in the underlying ad hoc network?

To answer this question, we consider a hybrid communication model. A Hybrid Communication Network has been introduced in different contexts~\cite{hybridWiredWireless,710380}.
To the best of our knowledge, we are the first ones that consider these types of networks for the purpose of finding paths in ad hoc networks.

At the core of our algorithm is a 2-localized Delaunay Graph of the ad hoc network. A 2-localized Delaunay Graph is related to the Delaunay triangulation, which was intruduced in~\cite{delaunay}. The advantage of using Delaunay graphs is that they are Euclidean $c$-spanners, which means that they contain a path for any pair of nodes of length at most $c$ times their Euclidean distance. The currently best known bound for $c$ is $1.998$ and was proven by Xia~\cite{xiaDelaunaySpanner}. Because wireless communication is only possible for limited distances, Delaunay graphs are not directly applicable to ad-hoc networks, but Delaunay graphs restricted to UDG edges, which are known as Restricted Delaunay Graphs \cite{localDelaunay}. Restricted Delaunay graphs are still hard to compute, so we will focus on the related 2-localized Delaunay Graph, which can be built in a constant number of rounds \cite{localDelaunay}. Based on that graph, one can use Chew's Algorithm \cite{boseUpperAndLowerBounds} to efficiently route messages if there are no radio holes.

Chew's Algorithm routes to nodes of triangles that intersect the direct line segment from the source node $s$ to the destination $t$. In cases where these triangulations are incomplete due to radio holes, routing strategies that only consider nodes of intersected triangles cannot be $o\left(n\right)$-competitive when $n$ is the number of nodes in the network \cite{10.1007/978-3-319-62389-4_6}. To avoid problems with radio holes, we calculate an abstraction of radio holes, i.e., we calculate the convex hulls of radio holes with the help of long-range links.
Computing the convex hull of a set of points is one of the most considered problems in computational geometry. There are publications dealing with distributed computation of convex hulls such as~\cite{bez1990distributed,rajsbaum2011some}.
Since we are interested in algorithms with polylogarithmic runtime, we use the parallel algorithm which computes the convex hull in time $\mathcal{O}\left(\log n\right)$ by~\cite{parallelConvexHullMiller}.
This algorithm makes use of a hypercube, which we build via the technique of pointer doubling.
This technique has been mentioned by Wyllie for the first time~\cite{Wyllie:1979:CPC:909252}.
Further, a parallel sorting algorithm in a hypercube is included in the preprocessing.
Batcher's Bitonic Sort has a deterministic parallel runtime of $\mathcal{O}(\log^ 2 n)$. Alternatively, the randomized algorithm of Reif and Valiant has an expected runtime of $\mathcal{O}(\log n)$.

Our approach of getting information about the 2-localized Delaunay graph inside a convex hull is based on calculating a Dominating Set of a special set of nodes inside this convex hull.
For the distributed computation of Dominating Sets, several algorithms have been proposed in the literature.
A popular algorithm has been introduced by Jia et al., which computes a $\mathcal{O}\left(\log \Delta\right)$ approximation of the smallest possible dominating set, where $\Delta$ denotes the degree of the network~\cite{jia2002efficient}.
Note that the computation of smallest dominating sets is proven to be NP complete. The algorithm requires $\mathcal{O}\left(\log n \cdot \log \Delta\right)$ communication rounds with high probability. \\

\section{Preliminaries} \label{section:preliminaries}
In this section, we introduce the preliminaries concerning the network topology and general results about routing in ad hoc networks.
In Section~\ref{section:adhocspanner}, we introduce the network topology for the ad hoc network in this work, the $2$-localized Delaunay Graph.
Moreover, we explain properties of the network structure.
Section~\ref{section:onlineRoutingTwoLocal} explains general results about routing in $2$-localized Delaunay Graphs.
\subsection{Spanner-Properties of the Ad Hoc Network} \label{section:adhocspanner}
In this work, we consider a $2$-localized Delaunay Graph \localDel{} as topology for the ad hoc network.
Before we give a formal definition of this topology, we introduce the Delaunay Graph.
Throughout this paper, we assume the set of nodes $V$ to be non-pathological, i.e., there are no three nodes on a line and no four nodes on a cycle.
Moreover, we assume that the coordinates of each node are unique and thus there are no two nodes on the same position.

\begin{definition}
	Let $\bigcirc\left(u,v,w\right)$ be the unique circle through the nodes $u, v$ and $w$ and $\triangle \left(u,v,w\right)$ be the triangle formed by the nodes $u,v$ and $w$.
	For any $V\subseteq \mathbb{R}^2$, the \emph{Delaunay Graph} $\mathrm{Del}\left(V\right)$ of $V$ contains all triangles $\triangle \left(u,v,w\right)$ for which $\bigcirc\left(u,v,w\right)$ does not contain any further node besides $u,v$ and $w$.
\end{definition}
\noindent
The Delaunay Graph does not restrict the length of an edge in any sense.
Hence, it is not a suitable graph structure for the ad hoc network as edges may exceed the transmission range of a node.
The $2$-localized Delaunay Graph is a structure that only allows edges which do not exceed the transmission range of a node.
Additionally, it can be constructed efficiently in a distributed manner.
The following definitions can be found in~\cite{localDelaunay}.

\begin{definition}
	A triangle $\triangle \left(u,v,w\right)$ satisfies the \emph{$k$-localized Delaunay property} if
	\begin{enumerate}
		\item all edges of $\triangle \left(u,v,w\right)$ have length at most $1$
		\item the interior disk of $\triangle \left(u,v,w\right)$ does not contain any node which can be reached within $k$ hops from $u,v$ or $w$ in UDG($V$).
	\end{enumerate}
	
\end{definition}

\begin{definition} \label{definition:klocalized}
	The \emph{$k$-localized Delaunay Graph} $LDel^k\left(V\right)$ consists of
	\begin{enumerate}
		\item all edges of $k$-localized triangles
		\item all edges $\left(u,v\right)$ for which the circle with diameter $\overline{uv}$ does not contain any further node $w \in V$ (\emph{Gabriel Edges})
	\end{enumerate}
\end{definition}
By choosing $k=2$, we obtain the $2$-localized Delaunay Graph which is also a planar graph.
Since $2$-localized Delaunay Graphs do not contain all edges of a corresponding Delaunay Graph, one cannot simply use routing strategies for Delaunay Graphs in our scenario.
We denote faces of the $2$-localized Delaunay Graph which are not triangles as \emph{holes}.
For the formal definition of holes, we distinguish between \emph{inner} and \emph{outer} holes.

\begin{definition} [Inner Hole] \label{definition:innerHole}
	Let $V \in \mathbb{R}^2$.
	An \emph{inner hole}  is a face of $LDel^2\left(V\right)$ with at least $4$ nodes.
\end{definition}

\begin{definition} [Outer Hole] \label{definition:outerHole}
	Let $V \in \mathbb{R}^2$. Furthermore, let $CH\left(V\right)$ be the set of all edges of the convex hull of $V$.
	Define $\overline{LDel^2}\left(V\right)$ to be the graph that contains all edges of \localDel{} and $CH\left(V\right)$.
	An \emph{outer hole} is a face in $\overline{LDel^2}\left(V\right)$ with at least $3$ nodes, that contains an edge $e \in CH\left(V\right)$ with $\|e\| > 1$.
\end{definition}
Nodes lying on the perimeter of a hole are called \emph{hole nodes}.
Note that the hole nodes of the same hole form a ring, i.e., each hole node  is adjacent to exactly two other hole nodes for each hole it is part of.

The choice of the $2$-localized Delaunay Graph as network topology is motivated by its \emph{spanner}-property.
We start with introducing spanner-properties of the original Delaunay Graph.
The Delaunay Graph $\mathrm{Del}\left(V\right)$ contains paths between every pair of nodes $v$ and $w$ of $V$ which are not longer than $c$ times their Euclidean distance.
We call these paths as follows:

\begin{definition}
	A \emph{path} $\left(v, \dots, w\right)$ between two nodes $v$ and $w$ in a geometric graph $G$ is a \emph{geometric $c$-spanning path} between $v$ and $w$, if its length is at most $c$ times the Euclidean distance between $w$ and $v$.
\end{definition}

Classes of graphs that contain such paths for every pair of nodes are called geometric $c$-spanners.

\begin{definition}
	A graph $G=\left(V,E\right)$ is called a \emph{geometric $c$-spanner}, if for all $v,w\in V$ there is a geometric $c$-spanning path $\left(v,\dots, w\right)$ in $G$.
\end{definition}

Delaunay Graphs are proven to be geometric $c$-spanners.
The currently best known bound on $c$ is $1.998$ and was proven by Xia~\cite{xiaDelaunaySpanner}.

\begin{theorem}
	There exists a path in a Delaunay Graph from node $s$ to $t$ of length less than $1.998 \cdot \|st\|$.
\end{theorem}
Xia argues that the bound of $1.998$ also relates to $2$-localized Delaunay Graphs~\cite{xiaDelaunaySpanner}.
However, these graphs are not spanners of the Euclidean metric but of the Unit Disk Graph.

\begin{theorem}
	In \localDel{} for $V \subset \mathbb{R}^2$, there exists a path between any pair of nodes $s$ and $t$ with length at most $1.998$ times their distance in $UDG(V)$.
\end{theorem}

\subsection{Online Routing in $2$-localized Delaunay Graphs} \label{section:onlineRoutingTwoLocal}
Finding paths in $2$-localized Delaunay Graphs that fulfill the spanning-property proven by Xia is only possible if knowledge about the entire graph is available.
Kuhn and Wattenhofer have proven that any routing strategy for $2$-localized Delaunay Graphs (or even for any type of graphs based on the Unit Disk links) which only considers the local
neighbors of each node cannot be constant-competitive for any constant $c$~\cite{Kuhn:2002:AOG:570810.570814}.
Bose et al.\ introduced the online routing strategy \emph{Chew's Algorithm} for Delaunay Graphs which only considers edges of triangles that are intersected by the direct line segment between source and destination~\cite{bez1990distributed}.

\begin{theorem}
	There exists an online routing strategy for Delaunay Graphs which finds a path between any source $s$ and target $t$ with length at most $5.9 \cdot  \|st\|$.
\end{theorem}

In case the source and the target node of the $2$-localized Delaunay Graph are \emph{visible} from each other, i.e., their direct line segment does not intersect any hole, Chew's Algorithm is also applicable.

\begin{theorem} \label{theorem:chewsAlgorithmViibleNodes}
	Let $s$ and $t$ be two visible nodes of a $2$-localized Delaunay Graph.
	Chew's Algorithm finds a path between $s$ and $t$ with length at most
	$5.9\|st\|$.
\end{theorem}

To be able to find constant-competitive paths between any pair of nodes in the $2$-localized Delaunay Graph, we take a look at results from computational geometry.
If we abstract from the underlying $2$-localized Delaunay Graph, our scenario is comparable to routing in polygonal domains.
These kinds of routing problems usually consider a starting point $s$ and a target point $t$ in the Euclidean plane.
The goal is to find a path in the plane from $s$ to $t$.
The challenging aspect of these problems is the presence of polygonal obstacles which avoid walking directly along the line segment $\overline{st}$.
In our scenario, these polygonal obstacles are radio holes.
De Berg et al.\ showed that it is enough to consider nodes of obstacle polygons for finding shortest paths in polygonal domains~\cite{computationalGeometryDeBerg}:

\begin{lemma} \label{lemma:polygonalPath}
	Any shortest path between $s$ and $t$ among a set $S$ of disjoint polygonal obstacles is a polygonal path whose inner nodes are nodes of $S$.
\end{lemma}

The usual procedure for finding shortest paths in polygonal domains is the computation of a Visibility Graph and applying a single source shortest path algorithm (e.g., the algorithm of Dijkstra)~\cite{computationalGeometryDeBerg}.
In the Visibility Graph $Vis\left(V\right)$ of a set of polygons, $V$ represents the set of corners of the polygon, and there is an edge $\{v,w\}$ in $Vis\left(V\right)$ if and only if a line can be drawn from $v$ to $w$ without crossing any polygon, i.e., $v$ is visible from $w$.

The combination of Theorem~\ref{theorem:chewsAlgorithmViibleNodes} and Lemma~\ref{lemma:polygonalPath} implies that a shortest path in the Visibility Graph of hole nodes of the $2$-localized Delaunay Graph (which would be the shortest possible geometric connection between the source and the target node) yields to a $5.9$-competitive path in the $2$-localized Delaunay Graph by applying Chew's Algorithm between every pair of consecutive nodes on the path.

\section{General Routing Protocol} \label{section:generalRoutingProtocol}

In this section, we introduce a general routing strategy for $2$-localized Delaunay Graphs by using long-range links to exchange information about locations and shapes of holes.
We improve this strategy in Section~\ref{section:routingProtocolConvexHulls} with respect to storage requirements and distributed computation time.
Throughout this section, we assume that the $2$-localized Delaunay Graph is already correct.
For more details about establishing a $2$-localized Delaunay Graph of ad hoc links, we refer the reader to Section~\ref{section:adHocProtocol}.
Furthermore, we assume for now that every node which is located on the perimeter of a hole stores a Visibility Graph of all hole nodes.
In our scenario, two hole nodes are visible from each other if their direct line segment does not intersect any hole of the $2$-localized Delaunay Graph.
In the following, we show that this helps us to compute competitive paths for the $2$-localized Delaunay Graph by introducing a routing strategy that finds these paths.
Section~\ref{section:concreteProtocol} later deals with aggregating the information needed for the Visibility Graphs. \\
The routing protocol works as follows:
A source node $s$ that wants to send data to a target node $t$ initially contacts $t$ via a long-range link to ask for $t's$ geographical position, i.e., a tuple of coordinates $(t_x,t_y)$.
$t$ responds with its position and $s$ afterwards sends its message via Chew's Algorithm towards $(t_x,t_y)$.
We distinguish two cases:
\begin{enumerate}
	\item The message reaches $t$ via Chew's Algorithm
	\item The message reaches a hole node $h_0$, i.e., the direct line segment $\overline{st}$ intersects a hole
\end{enumerate}
In case ($1$), we immediately obtain a $5.9$-competitive path from $s$ to $t$.
Otherwise, $h_0$ inserts $t$ into its Visibility Graph and applies a shortest path algorithm from itself to $t$.
The resulting shortest path $\left(h_0, h_1, h_2, \dots, h_k = t\right)$ is then used to transmit the message via ad hoc links.
By applying Chew's Algorithm, a path of length $5.9 \cdot \|h_0h_1\|$ is obtained.
After reaching $h_1$, the procedure is repeated until the message finally reaches $t$.
Let $p_{st}$ be the shortest path between $s$ and $t$ in the Visibility Graph.
In case $h_0$ lies on the shortest path between $s$ and $t$ in the Visibility Graph, the resulting path in the $2$-localized Delaunay Graph has length at most $5.9 \cdot \|p_{st}\|$ (based on Lemma~\ref{lemma:polygonalPath} and Chew's Algorithm).
Otherwise, the initial path to $h_0$ is a detour.
Nevertheless, it can be easily seen that the detour increases the competitiveness only by a constant factor.
As Chew's Algorithm did not reach $t$ but a node $h_0$, it follows that the path taken from $s$ to $h_0$ has length less or equal to $5.9\|st\|$ which is less or equal to $5.9$ times the shortest possible path between $s$ and $t$ in the $2$-localized Delaunay Graph.
Hence, the detour increases the competitive constant only by an additional factor of $3$.
Consequently, we obtain an $17.7$-competitive path between $s$ and $t$. \\
There are, however, some drawbacks with respect to the storage capacity required at hole nodes.
Unfortunately, the nodes on the perimeter of a radio hole potentially have to store a huge Visibility Graph.
In fact, it is possible to have a radio hole in \localDel{} with $\Theta\left(n\right)$ nodes on its perimeter.
Also, if $h$ denotes the number of nodes on the perimeter of a hole, then the Visibility Graphs may contain up to $\Theta\left(h^2\right)$ edges.
An idea to reduce the number of edges to $\mathcal{O}\left(h\right)$ is to not compute the entire Visibility Graph but only a Delaunay Graph of all nodes lying on different holes.
As Delaunay Graphs are planar graphs, this reduces the number of edges to $\mathcal{O}\left(h\right)$.
However, this also affects the obtained length of the paths.
Delaunay Graphs do not contain the shortest geometric connection between two nodes in general but a path which is $1.998$-competitive to such a path~\cite{xiaDelaunaySpanner}.
Hence, by using a Delaunay Graph instead of a Visibility Graph, we obtain a path length of $1.998 \cdot 17.7 \cdot \|p_{st} \| \leq 35.37 \cdot \|p_{st} \|$.

\section{Routing Protocol for Convex Hulls as Hole Abstractions} \label{section:routingProtocolConvexHulls}
In Section~\ref{section:generalRoutingProtocol}, we highlighted the advantage of a Visibility Graph or a Delaunay Graph of all hole nodes to find $c$-competitive paths in the ad hoc network.
Nevertheless, the storage requirements for each hole node are linear in the total number of hole nodes.
A natural question is how to reduce the number of nodes in the Visibility Graph even further while still being able to compute competitive paths.
In Section~\ref{section:convexHullSpace}, we show that considering only convex hulls of holes reduces the space requirements significantly in case the convex hulls of holes do not intersect.
Therefore, we assume for the rest of the paper that there is no pair of intersecting convex hulls of holes.
Moreover, we analyze in Section~\ref{section:convexHullsCompetitive} that considering only convex hulls still allows us to find competitive paths between almost all source-destination pairs.
Based on these observations, we introduce a $c$-competitive routing strategy similar to our protocol of Section~\ref{section:generalRoutingProtocol} which considers only hole nodes which lie on convex hulls of holes in Section~\ref{section:convexHullsRoutingProtocol}.
The mentioned protocol, however, cannot deal with specific cases in which both the source and the destination lie inside the same area inside of a convex hull.
A solution to these cases is introduced in Section~\ref{section:dominatingSet}.
\subsection{Space Reduction} \label{section:convexHullSpace}
We can obtain a further space reduction if we focus on locally convex hulls of the radio holes.

\begin{definition}
	Let $\left(v_1, v_2, \ldots,v_k,v_1\right)$ be a cycle of nodes in \\
	\localDel{} at the perimeter of some hole. We call ($v_{i_1}, v_{i_2},$ \\ $\ldots,v_{i_\ell}, v_{i_1}$) for some $1\le i_1<i_2<\ldots,i_\ell \le k$ a
	{\em locally convex hull} of that hole if (1) $\|v_{i_j} v_{i_{j+1}}\| \le 1$ for all $j \in \{1,\ldots,\ell\}$ (where $v_{i_{\ell+1}}=v_{i_1}$), and (2) there are no 3 consecutive nodes $u,v,w$ in that sequence where $\angle\left(u,v,w\right)\ge 180^\circ$ and $\|uw\| \le 1$.
\end{definition}
\noindent
For the locally convex hulls it can be shown:

\begin{lemma} \label{lemma:firstRandom}
	For any cycle $\left(v_1, v_2, \ldots,v_k,v_1\right)$ of hole nodes in \localDel{} that covers an area of size $A$, any locally convex hull of that cycle contains $\mathcal{O}\left(A\right)$ nodes.
\end{lemma}

\begin{proof}[Proof of Lemma~\ref{lemma:firstRandom}]
	Consider any locally convex hull $(v_{i_1}, v_{i_2},$ $\ldots,v_{i_\ell}, v_{i_1})$, and let $u,v,w$ be $3$ consecutive nodes in that sequence. I
	f $\angle\left(u,v,w\right) \ge 180^\circ$, then we know from the definition of the locally convex hull that $\|uw\|>1$. If $\angle(u,v,w)<180^\circ$, then $\|uw\|>1$ as well since otherwise $v$ would not be on the perimeter of the hole.
	This implies for the predecessor $p$ of $u$ and the successor $s$ of $w$ that $\|pv\|>1$ and $\|vs\|>1$.
	Also, there cannot exist any other node $x \in \{v_{i_1}, \ldots, v_{i_\ell}\}$ with $\|vx\|\le 1$ as otherwise we had a shortcut in the perimeter, meaning that $\left(v_1, v_2, \ldots,v_k,v_1\right)$ cannot be the perimeter of a hole.
	Hence, the unit cycle around each $v_{i_j}$ can contain at most 2 other nodes of the locally convex hull, which implies that $\ell =\mathcal{O}\left(A\right)$.
\end{proof}

\noindent
Hence, locally convex hulls contain a number of nodes that is independent of the total number of nodes in the system and only depends on the area covered by the hole.
A further reduction in the number of nodes can be achieved when only looking at the convex hull of a hole.

\begin{lemma} \label{lemma:randomSecond}
	For any cycle $\left(v_1, v_2, \ldots,v_k,v_1\right)$ of hole nodes in \localDel{} with a bounding box (i.e., the box of minimum size containing $v_1,\ldots,v_k$) of circumference $L$, the convex hull ($v_{i_1}, v_{i_2},\ldots,$ 
	$v_{i_\ell}, v_{i_1}$) of the cycle contains $\mathcal{O}\left(L\right)$ nodes.
\end{lemma}

\begin{proof}[Proof of Lemma~\ref{lemma:randomSecond}]
	Let $B$ be the bounding box of the cycle and $x$ be its center point. Let the points $w_{i_1},\ldots,w_{e_\ell}$ be the projections of $v_{i_1}, v_{i_2},\ldots,v_{i_\ell}$ from $x$ onto the boundary of $B$, i.e., the points where the ray from $x$ in the direction of $v_{i_j}$ intersects the boundary of $B$. As is easy to check, the $\ell_1$-distance of $w_{i_j}$ and $w_{i_{j+1}}$ on $B$ is at least as large as $\|v_{i_j} v_{i_{j+1}}\|$ for all $j$. Moreover, for any 3 consecutive points $u,v,w$ on the convex hull it must hold that $\|uw\|>1$. Hence, for any 3 consecutive points $u',v',w'$ on the projection of the convex hull onto $B$ it must hold that the $\ell_1$-distance of $u'$ and $w'$ is more than 1, which implies that the convex hull contains only $\mathcal{O}\left(L\right)$ nodes.
\end{proof}

\noindent
Since a bounding box of circumference $L$ may cover an area of size $\Theta\left(L^2\right)$, we may get another significant reduction in the number of nodes when switching from a locally convex hull to a convex hull.
All in all, by considering only convex hulls of holes, we achieve a significant reduction of the number of nodes contained in the Visibility Graph.
\subsection{$c$-competitive Paths via Convex Hulls} \label{section:convexHullsCompetitive}
In this section, we prove that considering only nodes of convex hulls of holes still allows us to find competitive paths in the $2$-localized Delaunay Graph.
For the moment, let us assume that both the source and the target of a routing request lie outside of any convex hull.
Moreover, we assume that the source and the target are not visible from each other as finding $c$-competitive paths for visible nodes can be found via Chew's Algorithm (see Section~\ref{section:onlineRoutingTwoLocal}).

\begin{lemma}\label{lemma:convexHullSufficient}
	The shortest path between any pair of non-visible nodes of the $2$-localized Delaunay Graph contains convex hull nodes.
\end{lemma}

\begin{proof}[Proof of Lemma~\ref{lemma:convexHullSufficient}]
	Let $s,t$ be two nodes of the $2$-localized Delaunay Graph,
	whose direct line segment intersects a hole.
	Starting from $s$, let $\ell$ be the first intersected line segment of the boundary
	of the intersected convex hull with endpoints $v,w$.
	We assume that the shortest path contains points of $\left(v,...,w\right)$.
	Else, the argumentation must be repeated with the neighboring
	edges of the convex hull.
	
	By contradiction, we assume that the shortest path from $s$ to $t$
	contains a point $p\in \left(v,...,w\right)$ from the interior of the convex hull,
	i.e., excluding $v,w$.
	Without loss of generality, we assume that the shortest path furthermore contains the point $w$ (The same holds
	for $v$.).
	
	Because of the triangle inequality, the following holds:
	$\|sw\| \leq \|sp\| + \|pw\|$
	And we know that
	$\|\left(x,y\right)\| \leq 1.998\cdot \|xy\|$
	holds for any two nodes of a Delaunay triangulation. \\
	\hfill \\
	\noindent
	Then: \\
	\begin{align*}
	\frac{\|\left(s,w\right)\|} {1.998}  &\leq \|sw\| \\ &\leq \|sp\| + \|pw\| \\ &\leq \|\left(s,p\right)\| + \|\left(p,w\right)\| \\
	&= \|\left(s,... ,p,... ,w\right)\|
	\end{align*}
	\hfill \\
	\noindent
	Hence, points of the interior of convex hulls cannot be chosen as a path along a convex hull node would be shorter.
\end{proof}

\noindent
Using this observation, we show that a Delaunay Graph of all convex hull nodes helps to find competitive
paths in the $2$-localized Delaunay Graph.
We define the \emph{Overlay Delaunay Graph} to be a Delaunay Graph that contains all convex hulls of holes and connects the nodes of different convex hulls in a Delaunay Graph.

The following theorem is a conclusion of the so far mentioned properties:

\begin{theorem} \label{theorem:convexHullPaths}
	Let $s$ and $t$ be two nodes of a $2$-localized Delaunay Graph that do not lie inside of any convex hull.
	Further, let $\left(s = c_0, c_1,\dots ,c_{\ell-1},c_\ell = t\right)$ be the shortest path in the Overlay Delaunay Graph via long-range links. Then we have
	\begin{enumerate}
		\item There is a $\left(1.998 \cdot \sum_{m=0}^{\ell-1}d_m\right)$-path in the $2$-localized Delaunay Graph from $s$ to $t$, where $ d_m:=\| c_mc_{m+1}\|$.
		\item By applying Chew's Algorithm, we obtain a \\
		$\left(5.9 \cdot \sum_{m=0}^{\ell-1}d_m\right)$-path in the $2$-localized Delaunay Graph from $s$ to $t$, where $ d_m:=\| c_mc_{m+1}\|$.
	\end{enumerate}
\end{theorem}

\noindent
Hence, we argue that this approach finds $c$-competitive paths from source $s$ to target $t$ in the $2$-localized Delaunay Graph. We will prove Theorem \ref{theorem:convexHullPaths}(1) with the following two lemmata:
\begin{lemma} \label{lemma:convexHullSpanningPath}
	Let $a$ and $b$ be visible nodes of different convex hulls. Then there is a $1.998 \cdot \|ab\|$-spanning path between them in the $2$-localized Delaunay Graph.
\end{lemma}

\begin{proof}[Proof of Lemma~\ref{lemma:convexHullSpanningPath}]
	Since the Delaunay Graph is a $1.998$-spanner of the complete Euclidean graph~\cite{xiaDelaunaySpanner} and the $2$-localized Delaunay Graph contains all edges of the original Delaunay Graph between a pair of visible nodes, there always exists a $1.998 \cdot \|ab\|$ path between two visible nodes $a$ and $b$ of two different convex hulls.
	This proves Lemma~\ref{lemma:convexHullSpanningPath}.
\end{proof}

\begin{lemma} \label{lemma:adjacentConvexHull}
	Let $a$ and $b$ be adjacent nodes on a convex hull, where $a\neq b$. Then there is a $1.998 \cdot \|ab\|$-spanning path in the $2$-localized Delaunay Graph between $a$ and $b$.
\end{lemma}

\begin{proof}[Proof of Lemma~\ref{lemma:adjacentConvexHull}]
	We use the observation of Xia that a $2$-localized Delaunay Graph is a $1.998$-spanner of the Unit Disk Graph.
	Thus there is a $1.998$-competitive path between the to convex hull nodes.
	This proves Lemma~\ref{lemma:adjacentConvexHull}.
\end{proof}
\noindent
And to prove Theorem~\ref{theorem:convexHullPaths}(2), we consider the following lemmata:

\begin{lemma} \label{lemma:visibleDifferentConvexHull}
	Let $a$ and $b$ be visible nodes of different convex hulls. Then there is a $ 5.9 \cdot  \|ab\|$-routing path between them in the $2$-localized Delaunay Graph.
\end{lemma}

\begin{proof}[Proof of Lemma~\ref{lemma:visibleDifferentConvexHull}]
	This fact follows immediately from \cite{boseUpperAndLowerBounds} and the fact
	that for two visible nodes $s$ and $t$, their direct line segment $\overline{st}$ intersects only
	triangles which are also part of the (standard) Delaunay Graph.
\end{proof}

\begin{lemma} \label{lemma:adjacentSameConvexHull}
	Let $a$ and $b$ be adjacent node on a convex hull, where $a\neq b$. Then there is a  $ 5.9 \cdot  \|ab\|$-routing path in the $2$-localized Delaunay Graph between $a$ and $b$.
\end{lemma}

\begin{proof}[Proof of Lemma~\ref{lemma:adjacentSameConvexHull}]
	The proof is the same as for Lemma~\ref{lemma:visibleDifferentConvexHull} because
	two adjacent convex hull nodes are due to the assumption of non-intersecting convex hulls per definition visible from each other.
\end{proof}

\noindent
Finally, we are able to prove Theorem~\ref{theorem:convexHullPaths}.

\begin{proof}[Proof of Theorem~\ref{theorem:convexHullPaths}(1)]
	Theorem~\ref{theorem:convexHullPaths}(1) follows immediately from Lemma~\ref{lemma:convexHullSpanningPath} and
	Lemma~\ref{lemma:adjacentConvexHull}.
\end{proof}

\begin{proof}[Proof of Theorem~\ref{theorem:convexHullPaths}(2)]
	Recall that there exists an online routing strategy for Delaunay Graphs which finds a path between any source $s$ and target $t$ with length at most $5.9 \cdot  \|st\|$.
	Furthermore, recall that the $2$-localized Delaunay Graph contains all edges of the original Delaunay Graph between any pair of visible nodes.
	Thus, there is a routing strategy from any convex hull node $a$ to any other convex hull node $b$ in cases $a$ and $b$ are nodes of different convex hulls with length at most $5.9 \cdot  \|ab\|$.
	This routing strategy can be applied to route in the $2$-localized Delaunay Graph between to adjacent convex hull nodes $a$ and $b$ as well. This is due to the fact, that the routing strategy chooses the path along triangles in the $2$-localized Delaunay Graph that are intersected by the line from $a$ to $b$.
	As $a$ and $b$ are visible from each other, these edges would also be part of the original Delaunay Graph.
	Thus, the routing strategy applied on the hybrid communication model gives a path of length at most $\left(5.9 \cdot \sum_{m=0}^{\ell-1}d_m\right)$.
	All in all, we obtain Theorem~\ref{theorem:convexHullPaths}.
\end{proof}

\subsection{Routing Protocol} \label{section:convexHullsRoutingProtocol}
This section deals with our routing protocol in the convex hull scenario.
Basically, we apply the same routing protocol as described in Section~\ref{section:generalRoutingProtocol} and instead of considering all hole nodes we only take those hole nodes which are also part of a convex hull into account.
To be precise, however, we have to consider more detailed cases concerning the positions of $s$ and $t$.
To investigate all different cases of the different geographical positions, we introduce \emph{bay areas}.
A bay area $H_A$ of a hole consists of the nodes and edges of the $2$-localized Delaunay Graph that are inside the convex hull and between two adjacent convex hull nodes.
For a visual intuition of bay areas, we refer the reader to Figure~\ref{figure:overview}.
The notion of bay areas allows us to formally describe each case we have to consider:

\begin{enumerate}
	\item $s$ and $t$ are outside of convex hulls
	\item $s$ or $t$ is inside of a convex hull
	\item $s$ and $t$ are inside different convex hulls
	\item $s$ and $t$ are inside the same convex hull but in different bay areas
	\item $s$ and $t$ are inside the same convex hull and in the same bay area.
\end{enumerate}

Case $1$ is solvable with few additional requirements to the routing protocol described in Section~\ref{section:generalRoutingProtocol}.
Cases $2-5$, however, need a more sophisticated solution and are postponed to Section~\ref{section:dominatingSet}.

For Case $1$, we assume for now that the following information is available:

\begin{enumerate}
	\item Each node located on the perimeter of a hole stores references to its two neighboring convex hull nodes
	\item All nodes lying on convex hulls of holes store an Overlay Delaunay Graph of all convex hull nodes
\end{enumerate}

The concrete routing protocol for Case $1$ works exactly as described in Section~\ref{section:generalRoutingProtocol}.
A node $s$ sends its message via Chew's Algorithm into the direction of $t$.
In case the message arrives at a hole node, it is directed to a convex hull node.
The convex hull node inserts $t$ into its Visibility Graph and applies a shortest path algorithm.
The resulting path is added to the message and used for forwarding the message in the ad hoc network.
Between any pair of nodes on the received path, Chew's Algorithm is applied.
Based on the results from Section~\ref{section:convexHullsCompetitive}, we obtain a $c$-competitive path in \localDel{}.

\subsection{Limitations of Convex Hulls} \label{section:dominatingSet}
The routing algorithm of Section~\ref{section:convexHullsRoutingProtocol} produces $c$-competitive paths between any pairs $(s,t)$, where the geographical coordinates of $s$ and $t$ are outside of convex hulls (Case $1$).
In this section, we concentrate on routing from $s$ to $t$, when their geographical coordinates fulfill the properties of Cases $2$-$5$. 
Here, we only provide the routing algorithm, where both $s$ and $t$ are in the same bay area, i.e., Case $5$.
It will be easy to see that an analogous routing can be executed for Cases $2$-$4$.\\
For computing $c$-competitive paths, we assume that a dominating set of all hole nodes in this bay area is known to each of these hole nodes.
A dominating set $DS$ of a graph $G=(V,E)$ is a subset of $V$ such that every node not in $DS$ is adjacent to at least one node of $DS$. To calculate $DS$, we refer to Section~\ref{section:incover}. \\

Recall that $\overline{s\,t}$ denotes the direct line segment between $s$ and $t$.
We define $S$ to be the first intersection point between $\overline{s\,t}$ and the hole boundary, from the direction of $s$.
Let $T$ be the analogous intersection point from the direction of $t$.
Let $P_1$ be the dominating set node with the shortest hop distance on the hole boundary to $S$ and $P_t$ the analogous dominating set node to $T$.
We denote $H_{s,t}$ to be the set of all hole nodes that are located in this bay area between $P_1$ and $P_t$.
We call the nodes of the convex hull of this set the \emph{extreme points} $\lbrace E_1,...,E_k\rbrace$.
We define $E_t$ to be the extreme point with the smallest index, where $\overline{E_t\, t}$ is visible to $t$. \\
The routing strategy works as follows:\\
$s$ executes Chew's Algorithm to send the message $m$ in the direction of $t$ until $m$
either arrives at $t$ (i.e., $s$ and $t$ are visible to each other) or at $P_1$.
If it reaches $P_1$, then $m$ is routed from $P_1$ to $E_1$, from $E_1$ to $E_2$,..., from $E_i$ to $E_t$, for $i=1,\dots,t$. Finally $m$ is routed from $E_t$ to $t$. All these routing steps are done with Chew's Algorithm. \\
Because Chew's Algorithm is $5.9$-competitive and the provided path by the algorithm contains in total $2+\vert E_{route}\vert$ direct lines, where $\vert E_{route}\vert$ denotes the number of extreme points that we route to, it is easy to see:
\begin{lemma}
	Let $s$ and $t$ be nodes with geographic coordinates in the same bay area, then the routing algorithm above provides a $c$-competitive routing path between $s$ and $t$ with $c=(2+\vert E_{route}\vert)\cdot 5.9$.
\end{lemma}

\section{Concrete Protocol} \label{section:concreteProtocol}

This section deals with collecting all information needed for the protocols described in Sections~\ref{section:generalRoutingProtocol} and \ref{section:routingProtocolConvexHulls} in a distributed manner.
The following issues have to be discussed:

\begin{enumerate}
	\item Distributed Construction of the $2$-localized Delaunay Graph
	\item Hole Detection
	\item Distributed computation of convex hulls
	\item Distribution of convex hull information to compute an Overlay Delaunay Graph
	\item Computation and Distribution of the Dominating Set of the hole ring in each bay area
\end{enumerate}

After points $1-5$ are solved, we are able to apply the routing strategies of Sections~\ref{section:generalRoutingProtocol} and \ref{section:routingProtocolConvexHulls} and obtain $c$-competitive paths in a completely distributed fashion. \\
The rest of the section is structured as follows:
Section~\ref{section:adHocProtocol} deals with the distributed construction of $2$-localized Delaunay Graph.
Afterwards, we describe a preprocessing strategy for the convex hull protocol, which transforms a ring of nodes into a hypercube.
The hypercube protocol is introduced in Section~\ref{section:hypercubeProtocol}.
Section~\ref{section:convexHullProtocol} introduces a protocol that computes convex hulls of all holes.
We continue with discussing hole detection, i.e., how nodes can detect if they are hole nodes in Section~\ref{section:holeDetection}.
Subsequently, the protocol for the distribution of convex hull information is introduced in Section~\ref{section:convexHullDistribution}.
Section~\ref{section:incover} deals with the computation of a Dominating Set along the hole ring in each bay area.

\subsection{Ad Hoc Network Construction} \label{section:adHocProtocol}
In the following, we discuss the distributed construction of a $2$-localized Delaunay Graph.
For the $2$-localized Delaunay Graph, we use the distributed protocol described in~\cite{localDelaunay}.
In their work, it is assumed that an initial (connected) Unit Disk Graph of all ad hoc links is given.
This can be trivially achieved if every node executes a WiFi-broadcast within its transmission range in a short setup-phase.
Afterwards, each node is aware of all nodes in its transmission range and we obtain a Unit Disk Graph.
As we cannot solve path finding in unconnected Unit Disk Graphs, we assume that the initial Unit Disk Graph is connected. \\
After all Unit Disk-links are known, the nodes execute the protocol of Li et al.\, which requires communication costs of $\mathcal{O}\left(n \log n\right)$ bits and only $\mathcal{O}(1)$ communication rounds~\cite{localDelaunay}.
The result is, to be precise, not a $2$-localized Delaunay Graph but a supergraph of it called \emph{Planar Localized Delaunay Graph}.
As each edge has a length of at most $1$ and the Planar Localized Delaunay Graph is a planar graph, our ideas of hole detection also work for these type of graphs.
For convenience, we restrict ourselves to $2$-localized Delaunay Graphs in the rest of this section.	\\

\subsection{Hypercube Protocol for a Ring of Nodes} \label{section:hypercubeProtocol}
In this section, we describe a procedure that establishes a hypercube topology out of a ring with $k$ nodes.
On the one hand, this protocol is a prerequisite for the convex hull protocol.
On the other hand, this protocol allows a fast hole detection, i.e., enables nodes to quickly distinguish the outer boundary from a hole.
More precisely, we execute the protocol both for holes and the outer boundary of the entire node set which are both connected in a ring topology.
For the ease of notation, we summarize nodes of the outer boundary and hole nodes as \textit{boundary nodes}.
Note that each node can locally detect whether it part of an inner or outer hole by checking whether it is part of a triangle with a missing edge due to the restriction of the edge length (see Definition~\ref{definition:klocalized}).
Each node $v$ which is part of the convex hull of the entire node set detects that there are two consecutive neighbors $v$ and $w$ in the clockwise ordering of $v's$ neighbors such that $\angle\left(u,v,w\right)\ge 180^\circ$. \\
Initially, each boundary node chooses a successor and a predecessor in each ring.
This can be achieved as follows:
Each boundary node sorts its boundary neighbors clockwise.
Afterwards, for every pair of consecutive nodes in the sorting (also for the last and the first node) the first node is chosen as predecessor and the second node is chosen as successor.
Now, every boundary is either oriented clockwise or counterclockwise.
More precisely, the outer boundary is oriented clockwise and each hole is oriented counterclockwise.
The orientation, however, is not important for the hypercube protocol but for the hole detection in Section~\ref{section:holeDetection}.\\
We proceed with the hypercube protocol by giving a definition of a hypercube.

\begin{definition}
	A $d$-dimensional \emph{hypercube} consists of $n$ nodes, where $n=2^d$, such that each node has a unique bitstring $(x_1, \dots, x_d)$ $\in \{0,1\}^d$ and there is an edge between two nodes if and only if their bitstring differs in only one bit.
	The decimal representation of a bitstring of a node $h$ is denoted as $id(v)$.
\end{definition}

For simplicity, we assume the number of nodes in the ring to be a power of two.
However, the techniques can be applied for an arbitrary number of nodes with a slight modification of the given protocol.
For the construction of the hypercube we use pointer jumping.
On the one hand, this technique enables us to build overlay edges for the hypercube fast and additionally it allows us
to elect a leader in $\mathcal{O}(\log k)$ communication rounds which is responsible for setting up the hypercube IDs.
The leader of the ring is the node with minimal ID.
The ID of a node $v$ is denoted as $id_v$.
In addition, we assign two values to each edge $e = \{u,v\}$, which is created by the pointer jumping protocol.
The first one, $\ell(e)$ defines the minimal ID of all ring nodes which are bridged by $e$, \emph{except} $id_u$.
The second value, $\mathrm{level}(e) = \log(b)$, where $b$ denotes the number of ring nodes between $u$ and $v$. \\
The pointer jumping is used as follows:
Let $v$ be a node of the hole ring and let $pred_0$ be its predecessor and $succ_0$ its successor on the ring.
In round $1$ of the protocol, $v$ introduces $succ_0$ to $pred_0$ to each other.
Thus $succ_0$ and $pred_0$ become adjacent nodes and an overlay edge $e = \lbrace pred_0, succ_0 \rbrace$ is established.
Further, $v$ assigns $\ell(e) = \min\{id_v,id_{succ_0}\}$ and $\mathrm{level}(e) = 0$.
As each node executes the protocol, $v$ also gets introduced two nodes in round $1$ which are denoted as $pred_1$ and $succ_1$.
In particular, in round $i$, each node $v$ of the hole ring introduces its predecessor $pred_{i-1}$ to its successor $succ_{i-1}$ and gets introduced $pred_i$ and $succ_i$.
The node $v$ that introduces $pred_{i-1}$ and $succ_{i-1}$ to each other also assigns $\ell(\lbrace pred_{i-1}, succ_{i-1} \rbrace) = \min\{\ell(\{pred_{i-1},v\}), \ell(\{v,succ_{i-1}\}) \}$  and $\mathrm{level}(e) = \mathrm{level}(\{pred_{i-1},v\}) +1$.  \\
With pointer jumping, the hop distance between any pair of nodes halves from round to round.
The protocol stops in a round $i$ in which $v$ gets introduced $succ_i$ and $pred_i$ and $\ell(\{pred_{i},v\}) = \ell(\{v, succ_i \})$.
At that point, each node (especially the leader itself) is locally aware of the minimal ID and hence knows the ID of the leader.
As the distance between any pair of nodes halves from round to round, this protocol requires $\mathcal{O}(\log k)$ communication rounds.
\noindent
For the purpose of being able to emulate a hypercube, we do not only need the additional overlay edges, but also hypercube IDs.
Recall that the node IDs of the hypercube are bitstrings of length $\log k$.
To distribute the hypercube IDs to the corresponding boundary nodes, the leader $v$ assigns for each hypercube edge $\{v, succ_i\}$ the binary representation of $\mathrm{level}(\{v, succ_i\})+1$ as ID to $succ_i$.
Each node that receives an ID from the leader repeats the ID distribution recursively, relative to its own ID.
As the diameter of a hypercube of $k$ nodes is $\mathcal{O}(\log k)$, the distribution of IDs requires $\mathcal{O}(\log k)$ communication rounds.
Eventually, the nodes of the ring form a hypercube and we are able to apply every protocol designed for hypercubes. \\
We summarize the results of this section in the following lemma:

\begin{lemma}
	A ring of $k$ nodes can be transformed into a hypercube in $\mathcal{O}(\log k)$ communication rounds.
	The number of required messages is in $\mathcal{O}(\log k)$ per node.
\end{lemma}

\subsection{Convex Hull Computation} \label{section:convexHullProtocol}
In the previous section, we presented the protocol to establish the hypercube of a ring of nodes.
We proceed with introducing a protocol that computes a convex hull of a ring of nodes that uses the hypercube protocol
as a subroutine.
For the convex hulls, we make use of the parallel algorithm of Miller which has been designed for hypercubes~\cite{parallelConvexHullMiller}.
The protocol requires $n$ sorted points.
More precisely, for hypercube nodes $h_1$ and $h_2$ with $id(h_1) < id(h_2)$, $h_1$ has to store a node of the ad hoc network with smaller ID than the node of the ad hoc network which is stored by $h_2$. \\
First, we apply the hypercube protocol of Section~\ref{section:hypercubeProtocol} and sort the points afterwards.
Sorting $n$ points in a hypercube can be done in $\mathcal{O}(\log n)$ communication rounds on expectation with the algorithm of Reif and Valiant~\cite{valiantSort}.
Upon termination, Miller's algorithm is applied which ensures
that each node of the ring knows every convex hull node and especially each convex hull node identifies itself as a convex hull node.
The following theorem follows:

\begin{theorem}
	Given a hole ring with $k$ nodes, the convex hull of this hole ring can be calculated in $\mathcal{O}(\log k)$ communication rounds on expectation.
\end{theorem}

\subsection{Hole Detection} \label{section:holeDetection}
In Section~\ref{section:hypercubeProtocol}, we have seen how to orient the cycle of nodes along the outer boundary clockwise and the ring of each hole counterclockwise.
However, boundary nodes locally cannot detect whether these cycles are oriented clockwise or counterclockwise and hence cannot decide whether they are located on the outer boundary or on a hole.
The idea to let nodes distinguish these cases is to sum up angles along each boundary into the direction of the orientation.
Let $v_1,v_2$ be a predecessor and a successor along a boundary.
In case walking from $v_1$ to $v_2$ requires a left turn, the angle between $v_1$ and $v_2$ is subtracted from the current sum.
Angles of right turns are added.
The result would be $360^\circ$ for the outer boundary and $-360^\circ$ for each hole~\cite{10.1007/978-3-319-72751-6_10}. \\
The summation along a boundary could be done by a token passing technique initiated by a leader.
This technique, however, requires a linear number of communication rounds for each cycle.
To improve the runtime, we sum angles in parallel to the hypercube protocol. in the following way:
In addition to the minimal ID, we also exchange the sum of angles with each edge of the pointer jumping procedure.
At the end, every node of the ring knows the sum of all angles along the boundary.
Hence, each node can decide whether it is a hole node in $\mathcal{O}(\log n)$ communication rounds. \\
For determining outer holes, we need a second run of convex hull computations along the outer boundary.
Note that outer holes are defined by an edge of the outer convex hull of the point set (see Definition~\ref{definition:outerHole}).
After the convex hull of the outer boundary has been computed, a second run is started between every pair of consecutive convex hull nodes whose distance exceeds the transmission range of a node.
All in all, we compute the convex hull of each hole and of the outer boundary to be able to distinguish the outer boundary and holes.
Afterwards we start a second run of convex hull computations for each outer hole determined by the convex hull of the outer boundary from the first run.
Finally, we have computed the convex hull of each hole in the network.

\subsection{Convex Hull Distribution} \label{section:convexHullDistribution}
In this section, we describe a strategy guaranteeing that all convex hull nodes are eventually connected in a clique via long-range links such that each convex hull node is locally able to compute an Overlay Delaunay Graph (see Section~\ref{section:convexHullsRoutingProtocol}).
The main observation of this section is that nodes of a convex hull locally cannot decide in which directions other holes are located (or even the existence of other holes).
Hence, we need to spread the information about convex hulls in the entire network.
A naive approach is to use a broadcast technique in which every convex hull node broadcasts itself together with the nodes which also belong to its convex hull in the network.
The runtime is limited by the diameter of the network (regarding hop-distance) which can be $\Theta\left(n\right)$ in $2$-localized Delaunay Graphs.
To achieve a faster distribution of broadcasts, we use an additional Overlay Network via the long-range links which only has a logarithmic diameter.
For doing so, we make use of a recently developed distributed protocol by Gmyr et al.\ which is designed for Hybrid Communication Networks~\cite{gmyr_et_al:LIPIcs:2017:7375}.
The protocol ensures that all nodes of the network are connected in a rooted tree via long-range links after $\mathcal{O}(\log^2 n)$ communication rounds.
The tree has a height of $\mathcal{O}(\log n)$ and a constant degree.
Consequently, the diameter of the tree is $\mathcal{O}(\log n)$.
As the diameter is only logarithmic, the tree allows us to distribute references of convex hull nodes in $\mathcal{O}(\log n)$ communication rounds in the following way:
Each convex hull node can direct its own reference both towards the root and into the subtree below itself.
The root redirects the reference into every other subtree.
This procedure avoids that nodes receive the same broadcast message multiple times.
The total runtime of this step is $\mathcal{O}(\log^2 n)$ as the tree has to be established initially. \\
So far, we have seen, that the $2$-localized Delaunay Graph, convex hulls of nodes and also the distribution of convex hull information can be achieved efficiently in $\mathcal{O}(\log^2 n)$ communication rounds which is dominated by the preprocessing protocol for the rooted tree.
The only part we left open until now is the routing protocol for nodes located in the same bay area.
The following Section~\ref{section:incover} deals with this problem.

\subsection{Dominating Set Protocol} \label{section:incover}
In cases where both the source and the target node of a routing request are located in the same bay area of a hole,
we have seen that a Dominating Set of the hole ring in that particular bay area helps to find $c$-competitive paths (see Section \ref{section:dominatingSet}). \\
The computation of a smallest possible dominating set, however, is proven to be NP-complete.
In this paper, we make use of the distributed dominating set protocol by Jia et al.\ which achieves
a $\mathcal{O}(\log \Delta)$-approximation of the smallest possible dominating set with high probability.
The parameter $\Delta$ denotes the degree of the network.
As we are computing dominating sets for hole rings in a bay area, $\Delta = 2$ in our scenario.
The protocol needs $\mathcal{O}(\log n \cdot \log \Delta)$ communication rounds with high probability.
Hence, we are able to compute a constant approximation of the smallest possible dominating set of a hole ring in a bay area in $\mathcal{O}(\log n)$ communication rounds with high probability.
We only have to take care that the protocol does not involve nodes of different bay areas.
However, convex hull nodes know that they are part of two bay areas and can take part in each dominating set protocol independently by only considering the neighbor of each particular bay area.

\section{Dynamic Scenario} \label{section:dynamicScenario}
In a real-world scenario with an Hybrid Communication Network consisting of smartphones, our assumption of immobile nodes is rather unrealistic.
In this section, we allow participants to move in each timestep while keeping the Unit Disk Graph connected.
Once the Overlay tree for fast exchange of convex hull information is built (see Section~\ref{section:convexHullDistribution}), we can obtain new convex hull information in only $\mathcal{O}(\log n)$ communication rounds.
As long as nodes do not leave and join the network, the Overlay tree remains valid as its structure does not depend on the position of the nodes.
Hence, the dominating runtime of $\mathcal{O}(\log^2 n)$ communication rounds for the tree is only required in an initial setup.
Therefore, a periodical re-execution of all protocols except the protocol for the distributed Overlay tree allows us to find competitive paths in a scenario where nodes are allowed to change their positions.

\section{Conclusion and Future Work}
In this paper, we investigated a Hybrid Communication Network consisting of a wireless ad hoc network, i.e., a $2$-localized Delaunay Graph, and an Overlay Network built via long-range links for the purpose of finding $c$-competitive paths in the ad hoc network in $\mathcal{O}(\log n)$ communication rounds.
Due to radio holes in the wireless ad hoc network, online routing strategies perform poor with respect to length of paths.
Therefore, we considered an Overlay Delaunay Graph consisting of the nodes of convex hulls of each radio hole.
We proved that knowledge about convex hulls suffices to find $c$-competitive paths in the $2$-localized Delaunay Graph.
Furthermore, we developed distributed protocols that detect holes in the ad hoc network, compute the convex hulls of each hole and distribute the information about convex hulls in the network such that each convex hull node locally stores an Overlay Delaunay Graph.
The Overlay Delaunay Graph enables convex hull nodes to compute competitive paths between nodes in the ad hoc network.
We proved that the total runtime of our protocols is $\mathcal{O}(\log^2 n)$ communication rounds.
When considering a dynamic scenario in which nodes are allowed to change their positions, we need $\mathcal{O}(\log^2 n)$ communication rounds for an initial setup but afterwards we are able to recompute the entire Overlay Network in only $\mathcal{O}(\log n)$ communication rounds.
Hence, our protocols are able to handle a dynamic scenario very efficiently. \\
Moreover, there are several challenging aspects which can be investigated in future research.
In this paper, we considered non-intersecting convex hulls.
The natural next step could be the design of routing strategies that can deal with finding competitive paths in areas of intersecting convex hulls.
Besides, we concluded our paper with a dynamic scenario in which nodes are allowed to move.
Our solution is to periodically recompute the entire Overlay Network.
This might not always be the best solution as usually nodes do not move arbitrarily fast.
Hence, a model with bounded movement speed could be investigated in which only parts of the Overlay Network have to be recomputed.
A further dynamic which could be considered, is joining and leaving nodes.
Lastly, our model does not tackle physical aspects of wireless communication.
Interesting aspects are for example wireless interference in crowded areas.
Also the signal power of wireless rays can be integrated into our theoretical model.

%%
%% Bibliography
%%

%% Either use bibtex (recommended), 
\bibliography{hybridBib}

\begin{thebibliography}{10}

\bibitem{computationalGeometryDeBerg}
Mark~de Berg, Otfried Cheong, Marc~van Kreveld, and Mark Overmars.
\newblock {\em Computational {G}eometry: {A}lgorithms and {A}pplications}.
\newblock Springer-Verlag TELOS, Santa Clara, CA, USA, 3rd edition, 2008.

\bibitem{bez1990distributed}
H.~E. Bez and J.~Edwards.
\newblock Distributed {A}lgorithm for the {P}lanar {C}onvex {H}ull {P}roblem.
\newblock {\em Computer-Aided Design}, 22(2):81--86, 1990.

\bibitem{boseUpperAndLowerBounds}
Nicolas Bonichon, Prosenjit Bose, Jean{-}Lou~De Carufel, Ljubomir Perkovic, and
  Andr{\'{e}} van Renssen.
\newblock Upper and {L}ower {B}ounds for {O}nline {R}outing on {D}elaunay
  {T}riangulations.
\newblock {\em Discrete {\&} Computational Geometry}, 58(2):482--504, 2017.

\bibitem{10.1007/3-540-40996-3_5}
Prosenjit Bose, Andrej Brodnik, Svante Carlsson, Erik~D. Demaine, Rudolf
  Fleischer, Alejandro L{\'o}pez-Ortiz, Pat Morin, and J.~Ian Munro.
\newblock Online routing in convex subdivisions.
\newblock In Gerhard Goos, Juris Hartmanis, Jan van Leeuwen, D.~T. Lee, and
  Shang-Hua Teng, editors, {\em Algorithms and Computation}, pages 47--59,
  Berlin, Heidelberg, 2000. Springer Berlin Heidelberg.

\bibitem{10.1007/978-3-319-62389-4_6}
Prosenjit Bose, Matias Korman, Andr{\'e} van Renssen, and Sander Verdonschot.
\newblock Constrained routing between non-visible vertices.
\newblock In Yixin Cao and Jianer Chen, editors, {\em Computing and
  Combinatorics}, pages 62--74, Cham, 2017. Springer International Publishing.

\bibitem{hybridWiredWireless}
G.~Cena, A.~Valenzano, and S.~Vitturi.
\newblock Hybrid wired/wireless {N}etworks for real-time {C}ommunications.
\newblock {\em IEEE Industrial Electronics Magazine}, 2(1):8--20, March 2008.

\bibitem{10.1007/978-3-319-72751-6_10}
Joshua~J. Daymude, Robert Gmyr, Andr{\'e}a~W. Richa, Christian Scheideler, and
  Thim Strothmann.
\newblock Improved leader election for self-organizing programmable matter.
\newblock In Antonio Fern{\'a}ndez~Anta, Tomasz Jurdzinski, Miguel~A. Mosteiro,
  and Yanyong Zhang, editors, {\em Algorithms for Sensor Systems}, pages
  127--140, Cham, 2017. Springer International Publishing.

\bibitem{delaunay}
Boris Delaunay.
\newblock {Sur la sph\`{e}re vide. A la {M}\'{e}moire de {G}eorges
  {V}orono\"{i}}.
\newblock {\em Bulletin de l'Acad\'{e}mie des Sciences de l'URSS}, 6:793--800,
  1934.

\bibitem{gmyr_et_al:LIPIcs:2017:7375}
Robert Gmyr, Kristian Hinnenthal, Christian Scheideler, and Christian Sohler.
\newblock {Distributed Monitoring of Network Properties: The Power of Hybrid
  Networks}.
\newblock In Ioannis Chatzigiannakis, Piotr Indyk, Fabian Kuhn, and Anca
  Muscholl, editors, {\em 44th International Colloquium on Automata, Languages,
  and Programming (ICALP 2017)}, volume~80 of {\em Leibniz International
  Proceedings in Informatics (LIPIcs)}, pages 137:1--137:15, Dagstuhl, Germany,
  2017. Schloss Dagstuhl--Leibniz-Zentrum fuer Informatik.

\bibitem{gupta2011review}
Anuj~K Gupta, Harsh Sadawarti, and Anil~K Verma.
\newblock Review of various routing protocols for manets.
\newblock {\em International Journal of Information and Electronics
  Engineering}, 1(3):251, 2011.

\bibitem{jia2002efficient}
Lujun Jia, Rajmohan Rajaraman, and Torsten Suel.
\newblock An {E}fficient {D}istributed {A}lgorithm for {C}onstructing {S}mall
  {D}ominating {S}ets.
\newblock {\em Distributed Computing}, 15(4):193--205, 2002.

\bibitem{Kuhn:2002:AOG:570810.570814}
Fabian Kuhn, Rogert Wattenhofer, and Aaron Zollinger.
\newblock Asymptotically optimal geometric mobile ad-hoc routing.
\newblock In {\em Proceedings of the 6th International Workshop on Discrete
  Algorithms and Methods for Mobile Computing and Communications}, DIALM '02,
  pages 24--33, New York, NY, USA, 2002. ACM.

\bibitem{kuhn2003worst}
Fabian Kuhn, Rogert Wattenhofer, and Aaron Zollinger.
\newblock Worst-{C}ase {O}ptimal and {A}verage-case {E}fficient {G}eometric
  {A}d-hoc {R}outing.
\newblock In {\em Proceedings of the 4th ACM International Symposium on Mobile
  Ad Hoc Networking \&Amp; Computing}, MobiHoc '03, pages 267--278, New York,
  NY, USA, 2003. ACM.

\bibitem{7236907}
S.~Li, W.~Zeng, D.~Zhou, X.~Gu, and J.~Gao.
\newblock Compact conformal map for greedy routing in wireless mobile sensor
  networks.
\newblock {\em IEEE Transactions on Mobile Computing}, 15(7):1632--1646, July
  2016.

\bibitem{localDelaunay}
Xiang-Yang Li, G.~Calinescu, and Peng-Jun Wan.
\newblock Distributed {C}onstruction of a {P}lanar {S}panner and {R}outing for
  {A}d {H}oc {W}ireless {N}etworks.
\newblock In {\em Proceedings of the 21st Annual Joint Conference of the IEEE
  Computer and Communications Societies}, volume~3, pages 1268--1277 vol.3, New
  York, NY, USA, 2002. {IEEE} Press.

\bibitem{lumelsky1987algorithmic}
Vladimir~J Lumelsky.
\newblock Algorithmic and {C}omplexity {I}ssues of {R}obot {M}otion in an
  {U}ncertain {E}nvironment.
\newblock {\em Journal of Complexity}, 3(2):146--182, 1987.

\bibitem{parallelConvexHullMiller}
R.~Miller and Q.~F. Stout.
\newblock Efficient {P}arallel {C}onvex {H}ull {A}lgorithms.
\newblock {\em IEEE Transactions on Computers}, 37(12):1605--1618, Dec 1988.

\bibitem{710380}
Y.~S.~N. Murty.
\newblock Hybrid {C}ommunication {N}etworks for {P}ower {U}tilities.
\newblock In {\em Power Quality '98}, pages 239--242, New York, NY, USA, Jun
  1998. {IEEE} press.

\bibitem{rajsbaum2011some}
Sergio Rajsbaum and Jorge Urrutia.
\newblock Some {P}roblems in {D}istributed {C}omputational {G}eometry.
\newblock {\em Theoretical Computer Science}, 412(41):5760--5770, 2011.

\bibitem{rao1993robot}
Nagewara~SV Rao, Srikumar Kareti, Weimin Shi, and S~Sitharama Iyengar.
\newblock Robot {N}avigation in {U}nknown {T}errains: {I}ntroductory {S}urvey
  of non-heuristic {A}lgorithms.
\newblock Technical report, Oak Ridge National Lab., TN (United States), 1993.

\bibitem{valiantSort}
John~H. Reif and Leslie~G. Valiant.
\newblock A {L}ogarithmic {T}ime {S}ort for {L}inear {S}ize {N}etworks.
\newblock {\em Journal of the ACM (JACM)}, 34(1):60--76, January 1987.

\bibitem{ruhrup2006online}
Stefan R{\"u}hrup and Christian Schindelhauer.
\newblock Online multi-path routing in a maze.
\newblock In Tetsuo Asano, editor, {\em Algorithms and Computation}, pages
  650--659, Berlin, Heidelberg, 2006. Springer Berlin Heidelberg.

\bibitem{Wyllie:1979:CPC:909252}
James~Christopher Wyllie.
\newblock {\em The Complexity of Parallel Computations}.
\newblock PhD thesis, Ithaca, NY, USA, 1979.

\bibitem{xiaDelaunaySpanner}
Ge~Xia.
\newblock The {S}tretch {F}actor of the {D}elaunay {T}riangulation is less than
  1.998.
\newblock {\em {SIAM} {Journal on Computing}}, 42(4):1620--1659, 2013.

\end{thebibliography}

%% .. or use the thebibliography environment explicitely
\newpage

\appendix

\end{document}